\documentclass[a4paper,12pt]{article}
\usepackage{amsmath,amsthm,amssymb,setspace}
\usepackage{newtxtext,newtxmath}

\usepackage{hyperref}

\usepackage[margin=1.00in]{geometry}
\usepackage{subfigure}
\usepackage{bbm}
\usepackage{array,arydshln}
\usepackage{comment}
\usepackage[round]{natbib}

\usepackage[dvipdfmx]{graphicx,color}
\let\sum\relax

\DeclareSymbolFont{CMlargesymbols}{OMX}{cmex}{m}{n}
\DeclareMathSymbol{\sum}{\mathop}{CMlargesymbols}{"50}

\newtheorem{theorem}{Theorem}
\newtheorem{prop}{Proposition}

\newtheorem{lem}{Lemma}
\newtheorem{cor}{Corollary}
\theoremstyle{definition}

\newtheorem{as}{Assumption}

\begin{document}

	\title{
	A General Impossibility Theorem on Pareto Efficiency and Bayesian Incentive Compatibility\footnote{
	We are thankful to 
Mat\'{i}as N\'{u}\~{n}ez
for the fruitful
discussions and useful comments.
Financial supports by Investissements d'Avenir, ANR-11-IDEX-0003/Labex Ecodec/ANR-11-LABX-0047 
and
PHC Sakura program, project number 45153XK,
are gratefully acknowledged.
}
 
	}
	
\author{
	Kazuya Kikuchi\thanks{
	{Tokyo University of Foreign Studies. E-mail: \texttt{kazuya.kikuchi68@gmail.com}}.}
	\and
	Yukio Koriyama\thanks{
	CREST, Ecole Polytechnique, Institut Polytechnique de Paris. 
	E-mail: \texttt{yukio.koriyama@polytechnique.edu}.}}  
	
	\date{This version: March, 2023}
	
	\maketitle
	
	\begin{abstract}
This paper studies a general class of social choice problems in which agents' payoff functions (or types) are privately observable random variables, and monetary transfers are not available. We consider cardinal social choice functions which may respond to agents' preference intensities as well as preference rankings. We show that a social choice function is ex ante Pareto efficient and Bayesian incentive compatible if and only if it is dictatorial. The result holds for arbitrary numbers of agents and alternatives, and under a fairly weak assumption on the joint distribution of types, which allows for arbitrary correlations and asymmetries. 
	\end{abstract}

\section{Introduction}

	We consider a general class of social choice problems in which agents’ payoffs are privately observed random variables and monetary transfer is not available. Payoffs are drawn from a continuous space, so that the aggregation of the cardinal preferences is at stake. We prove an impossibility theorem: a social choice function is ex ante Pareto efficient and Bayesian incentive compatible if and only if it is dictatorial, under fairly general conditions.

    Our main contribution lies in the generality of the conditions under which the impossibility holds. We provide a proof of the theorem for an arbitrary number of agents and alternatives, under a mild assumption about the joint distribution of preference types. In particular, our theorem does not hinge on either independence or symmetry assumptions.
    
	We would like to emphasize that allowing for correlation is relevant in many applications. In voting in general, it is common for voters' preferences for each alternative to be correlated. Our theorem provides insight not only for normative analysis in which assumptions such as independence or symmetry of the preferences are imposed, but also for positive analysis in which presence of correlations in the joint distribution is relevant. 
 
	Our proof proceeds by highlighting the incentive for agents to exaggerate the intensity of their preferences on the best and the worst alternatives, which we call \textit{extremization}. 
    Our result suggests that the incentive that agents are willing to send an extreme message about their best and the worst alternatives is robust, and the idea that such a motivation would lead to an impossibility theorem is formalized in the proof.
 
    Impossibility results have been obtained under various conditions in the literature \citep{BORGERS20092057,EHLERS202031}. However, to the best of our knowledge, existing models are either limited with respect to the number of agents, the number of alternatives, or the conditions on the joint distribution of preferences. We discuss our contribution to the literature in more detail below. 

    \subsection{Related literature}

Our result concerns the efficiency and incentive properties of \textit{cardinal} social choice rules in a Bayesian environment without monetary transfers. Cardinal rules take into account the preference intensities of agents, unlike \textit{ordinal} rules which are based only on the preference rankings.

Previous studies on this topic have already demonstrated the tradeoff between ex ante efficiency and incentives under cardinal rules. 
The result closest to ours is \citet[
	Theorem 8]{EHLERS202031} which shows that if each agent's belief about the other agents' types is independent of his own type, then any nondictatorial weighted utilitarian rule is not Bayesian incentive compatible. 
\cite{BORGERS20092057} derives the same conclusion in a setting with two agents who have opposite preference rankings over three alternatives, and whose von-Neumann Morgenstern (vNM) utilities from the middle-ranked alternative are independent and identically distributed random variables. Our motivation in this paper is to establish the impossibility result under a minimum assumption. We only assume that the joint distribution of agents' types has a density with full support. In particular, we allow for arbitrary correlations of types.

For ordinal rules, the well-known Gibbard-Satterthwaite theorem (\citealt{gibbard1973manipulation} and \citealt{satterthwaite1975strategy}) shows that strategy-proofness implies dictatorship, if there are at least three alternatives. \cite{majumdar2004ordinally} considers the concept of ordinal Bayesian incentive compatibility (OBIC) which requires that truth-telling is a best response for each agent with respect to all vNM utility functions consistent with the agent's ordinal ranking. They show that with three or more alternatives, there is a generic set of independent belief systems for which an ordinal rule is OBIC only if it is dictatorial, demonstrating the robustness of the Gibbard-Satterthwaite theorem. However, for a certain class of positively correlated belief systems,
\cite{bhargava2015incentive} provides a sufficient condition (much weaker than dictatorship) for an ordinal rule to be locally robustly OBIC, in the sense that the rule is OBIC in some neighborhood of the given belief system.

Other works in the literature investigate optimal rules among cardinal rules, subject to incentive constraints. \cite{AzrieliKim2014} considers the case with two alternatives and independent types, and show that ex ante efficient rules and interim efficient rules subject to Bayesian incentive compatibility are characterized by certain classes of (ordinal) weighted majority rules. \cite{schmitz2012sub} also studies the case with two alternatives. They assume that the distribution of types is symmetric with respect to both the alternatives and agents, but allow the types to be correlated. Then some weak majority rule, in which the two alternatives are chosen with equal probabilities unless either alternative obtains a sufficiently large support, is optimal among strategy-proof rules. Moreover, the standard majority rule is optimal among Bayesian incentive compatible rules if types are independent, but there can be some superior rule if types are correlated.

These results suggest that an incentive compatible rule cannot effectively utilize the cardinal information of preferences. Indeed, \cite{EHLERS202031} shows that if a social choice rule satisfies some continuity condition in addition to Bayesian incentive compatibility, then it must be vNM-ordinal in expectation, meaning that the rule responds to a change in an agent's type only if his expected-utility ranking of alternatives changes.
\cite{kim2017ordinal} shows that there is a Bayesian incentive compatible vNM-ordinal rule that achieves a higher utilitarian welfare than any ordinal rule.

Finally, \cite{jackson2007overcoming} shows that for a general Bayesian collective decision problem, any ex ante Pareto efficient social choice rule can be approximately implemented by the linking mechanism, defined on a set of independent copies of the original decision problem.

	\section{The impossibility theorem}
	
	Let $I$ be the set of agents with $|I|=n\geq2$, and $X$ the set of alternatives with $|X|=m\geq2$, both finite. If alternative $x\in X$ is chosen, agent $i\in I$ obtains a random payoff $U_i^x$ taking values in $[0,1]$. The payoff vector $U_i:=(U_i^x)_{x\in X}\in[0,1]^m$ is called the \textit{type} of the agent. We assume the following:
	\begin{as}
    	The joint distribution of the type profile $U:=(U_i)_{i\in I}$ is absolutely continuous (with respect to Lebesgue measure) and has full support over $[0,1]^{mn}$.
	\label{as:dist}
	\end{as}

	A \textit{social choice function} (SCF) is a Borel measurable function $f:[0,1]^{mn}\to\Delta(X)$, where $\Delta(X):=\{(p^x)_{x\in X}\in[0,1]^m|\sum p^x=1\}$ is the set of probability distributions on $X$. Given the realization $u$ of the type profile $U$, the SCF chooses an alternative according to the distribution $f(u)$. For convenience, we denote by $\tilde{f}(u)\in X$ the (possibly random) alternative chosen by $f$ when the type profile is $u$. The expected payoff for agent $i$ under the SCF $f$ is then \[
    \pi_i(f):=\mathbf{E}\left[U_i^{\tilde{f}(U)}\right]. 
	\]

	An SCF $f$ is \textit{ex ante Pareto efficient} if there is no SCF $g$ such that for all agents $i$, $\pi_i(g)\geq\pi_i(f)$, with strict inequality for some $i$. An SCF $f$ is \textit{dictatorial} if there is an agent $i$ such that $U_i^{\tilde{f}(U)}=\max_{x}U_i^x$ almost surely.

	The \textit{direct mechanism} for an SCF $f$ is the Bayesian game played by the $n$ agents in which: each agent's type is private information; a \textit{strategy} for agent $i$ is a Borel measurable function $\sigma_i=(\sigma_i^x(\cdot))_{x\in X}:[0,1]^m\to[0,1]^m$, under which the agent reports that his type is $\sigma_i(u_i)=(\sigma_i^x(u_i))_{x\in X}\in[0,1]^m$ if his true type is $u_i$; given the profile of reported types $\sigma(u)=(\sigma_i(u_i))_{i\in I}$, the outcome of the mechanism is the (possibly random) alternative $\tilde{f}(\sigma(u))$, and the payoff for agent $i$ is $u_i^{\tilde{f}(\sigma(u))}$. The expected payoff for agent $i$ under the strategy profile $\sigma$ is 
    \[
    \pi_i^f(\sigma):=\pi_i(f\circ\sigma)=\mathbf{E}\left[U_i^{\tilde{f}(\sigma(U))}\right].
    \]
    We also introduce the notation 
 \[
 \pi_i^f(v_i,\sigma_{-i}|u_i):=\mathbf{E}\left[u_i^{\tilde{f}(v_i,\sigma_{-i}(U_{-i}))}\Big|U_i=u_i\right]
 \]
 for the conditional expected payoff given the true type $u_i$, if the agent reports that his type is $v_i$ while the other agents play the strategies $\sigma_{-i}=(\sigma_j)_{j\neq i}$.

	In the direct mechanism, a strategy profile $\sigma=(\sigma_i)_{i\in I}$ is a \textit{Bayesian Nash equilibrium} (BNE) if $\pi_i^f(\sigma)\geq \pi_i^f(\sigma_i^\prime,\sigma_{-i})$ for all agents $i$ and all strategies $\sigma_i^\prime$. An equivalent definition is the following: $\sigma$ is a BNE if $\pi_i^f(u_i,\sigma_{-i}|u_i)\geq \pi_i^f(v_i,\sigma_{-i}|u_i)$ for all agents $i$, almost all types $u_i\in[0,1]^m$ and all $v_i\in[0,1]^m$. We denote by $\tau_i$ the \textit{truth-telling strategy} for agent $i$: $\tau_i(u_i)\equiv u_i$. The SCF $f$ is \textit{Bayesian incentive compatible} if the profile of truth-telling strategies $\tau=(\tau_i)_{i\in I}$ is a BNE.
	
	\begin{theorem}
	A social choice function is ex ante Pareto efficient and Bayesian incentive compatible if and only if it is dictatorial.\footnote{A restricted version of Theorem 1 for the case with two alternatives has appeared as Proposition 3 in the working-paper version of \cite{kikuchi2022winner}, available at https://arxiv.org/abs/2206.09574.}
	\label{thm:main}
	\end{theorem}

	\section{Proof of the theorem}
	
	To prove Theorem \ref{thm:main}, we first characterize ex ante Pareto efficient SCFs as weighted utilitarian SCFs (Proposition \ref{prop:eff}). We then show that the only weighted utilitarian SCFs that are Bayesian incentive compatible are dictatorial SCFs (Proposition \ref{prop:bic}).

	\bigskip

    Let $\Delta(I):=\{(\lambda_i)_{i\in I}\in[0,1]^n|\sum\lambda_i=1\}$ be the set of weight assignments to the agents, with sum normalized to one.
	We call an SCF $f$ \textit{weighted utilitarian} if
	there is $\lambda\in\Delta(I)$ such that $\sum_i\lambda_iU_i^{\tilde{f}(U)}=\max_{x}\sum_i \lambda_iU_i^x$ almost surely. Note that a dictatorial SCF is a weighted utilitarian SCF that gives the whole weight (1) to one agent.
	
	\begin{prop}
	A social choice function is ex ante Pareto efficient if and only if it is weighted utilitarian.
	\label{prop:eff}
	\end{prop}

	\begin{proof}
    Given an SCF $f$, denote by $\pi(f):=(\pi_i(f))_{i\in I}\in[0,1]^n$ the profile of expected payoffs.
	Let $F$ be the set of all SCFs, and $\pi(F)=\{\pi(f)|f\in F\}$ the set of all expected payoff profiles induced by SCFs. The set $\pi(F)$ is a convex subset of $[0,1]^n$, since $F$ permits SCFs that randomize over the alternatives.
	
	\textit{``Only if'' part.} Let $f$ be an ex ante Pareto efficient SCF. Since $\pi(F)$ is a convex set in $[0,1]^n$, by \citet[
	Proposition 16.E.2]{MWG1995}, there is $\lambda\in\Delta(I)$ such that the expected payoff profile $p=\pi(f)$ is a solution to the following maximization problem:
	\begin{equation}
	\max_{p\in \pi(F)}\lambda\cdot p.
	\label{eq:max_p}
	\end{equation}
	Thus, $f$ must be a solution to the maximization problem:
	\begin{equation}
	\max_{f\in F}\lambda\cdot \pi(f).
	\label{eq:max_f}
	\end{equation}
	By linearity of expectation, the objective function in (\ref{eq:max_f}) equals $\mathbf{E}\left[\sum_i\lambda_iU_i^{\tilde{f}(U)}\right]$, which is maximized if and only if $f$ is a $\lambda$-weighted utilitarian SCF. This proves the ``only if'' part. 
	
	\textit{Uniqueness of the solution to (\ref{eq:max_p}).}
	The above argument shows that if $p$ is a solution of (\ref{eq:max_p}), then $p=\pi(f)$ for some $\lambda$-weighted utilitarian $f$. There are multiple $\lambda$-weighted utilitarian SCFs, as there can be multiple alternatives $x$ maximizing the $\lambda$-weighted utilitarian sum $\sum_i\lambda_iU_i^x$. However, absolute continuity implies that such ties occur with probability 0. This implies that all $\lambda$-weighted utilitarian SCFs coincide almost surely, and hence induce the same expected payoff vector $\pi(f)$. Therefore, the solution $p$ of (\ref{eq:max_p}) is unique.
	
	\textit{``If'' part.} Suppose on the contrary that there is $\lambda\in\Delta(I)$ for which a $\lambda$-weighted utilitarian SCF $f$ is not ex ante Pareto efficient. Let $p=\pi(f)$ be the expected payoff profile under this SCF. There is a Pareto-dominating expected payoff profile $q\in\pi(F)$ with $q\neq p$ such that $q_i\geq p_i$ for all $i$. Then $\lambda\cdot q\geq\lambda\cdot p$, which contradicts the fact that $p$ is the unique solution to (\ref{eq:max_p}).
	\end{proof}
	
	We next show that the dictatorial SCFs are the only weighted utilitarian SCFs that are Bayesian incentive compatible.
	To do this, we first prove Lemma \ref{lem:dom} below, which shows that under a nondictatorial weighted utilitarian SCF, agents have a strategic incentive to report the extreme payoff values (0 and 1) for their worst and best alternatives, even when the true payoffs from those alternatives are intermediate values.

	Formally, we define \textit{extremization} as the function $\hat{(\cdot)}$ that transforms each strategy $\sigma_i$ into the strategy $\hat{\sigma}_i$ given by,
	for each $u_i\in[0,1]^m$ and $x\in X$,
	\[
	\hat{\sigma}_i^x(u_i):=\begin{cases}
	0 & \text{if $u_i^x=\min_{y} u_i^y<\max_{y}u_i^y$}\\
	1 & \text{if $u_i^x=\max_{y} u_i^y>\min_{y}u_i^y$}\\
	\sigma_i(u_i) & \text{otherwise.}
	\end{cases}
	\]
	We also need the following definitions.
	We say that two strategies $\sigma_i$ and $\sigma_i^\prime$ are \textit{equivalent} if $\sigma_i(U_i)=\sigma_i^\prime(U_i)$ almost surely. 
	A strategy subprofile $\sigma_{-i}=(\sigma_j)_{j\neq i}$ for the agents except $i$ is called \textit{regular} if the joint distribution of reported types $(\sigma_j(U_j))_{j\neq i}$ has full support over $[0,1]^{(n-1)m}$ and is absolutely continuous. For example, the truth-telling strategy subprofile $\tau_{-i}$ is regular since $(\tau_j(U_j))_{j\neq i}=(U_j)_{j\neq i}$ has full support and is absolutely continuous by Assumption \ref{as:dist}. 
	Weak dominance of strategies is defined in the usual way: strategy $\sigma_i$ \textit{weakly dominates} strategy $\sigma_i^\prime$ if  $\pi_i(\sigma_i,\sigma_{-i})\geq\pi_i(\sigma_i^\prime,\sigma_{-i})$ for all $\sigma_{-i}$, and $\pi_i(\sigma_i,\sigma_{-i})>\pi_i(\sigma_i^\prime,\sigma_{-i})$ for some $\sigma_{-i}$.

	\begin{lem} 
	Let $f$ be a nondictatorial weighted utilitarian SCF with weights $(\lambda_i)_{i\in I}$. Then there is at least one agent $i$ with $0<\lambda_i\leq1/2$. For each such agent $i$, every strategy $\sigma_i$ that is not equivalent to its extremization $\hat{\sigma}_i$ is weakly dominated by $\hat{\sigma}_i$; specifically, $\pi_i^f(\hat{\sigma}_i,\sigma_{-i})\geq\pi_i^f(\sigma_i,\sigma_{-i})$ for all $\sigma_{-i}$, and $\pi_i^f(\hat{\sigma}_i,\sigma_{-i})>\pi_i^f(\sigma_i,\sigma_{-i})$ for all regular $\sigma_{-i}$. 
	\label{lem:dom}
	\end{lem}

	\begin{proof}
    Since the SCF $f$ is nondictatorial, there are at least two agents with positive weights, one of which cannot exceed one-half of the total weight (1). Thus there is an agent $i$ satisfying $0<\lambda_i\leq1/2$. Fix such an agent $i$ and a strategy $\sigma_i$ not equivalent to its extremization $\hat{\sigma}_i$.

    Clearly, under-reporting the payoff from the worst alternative or over-reporting the payoff from the best alternative never makes the agent worse off. Thus, given any type $u_i$ and any strategy subprofile $\sigma_{-i}$, if $a$ and $b$ are the best and worst alternatives for the type $u_i$, the conditional expected payoff $\pi_i^f(v_i,\sigma_{-i}|u_i)$ is nondecreasing in $v_i^a\in[0,1]$ and nonincreasing in $v_i^b\in[0,1]$. In particular,
	\begin{equation}
	    \pi_i^f(\hat{\sigma}_i(u_i),\sigma_{-i}|u_i)\geq
	    \pi_i^f({\sigma}_i(u_i),\sigma_{-i}|u_i)\text{ for all $u_i$ and $\sigma_{-i}$}.
	    \label{eq:ineq}
	\end{equation}
	This implies the weak inequality of unconditional expected payoffs: $\pi_i^f(\hat{\sigma}_i,\sigma_{-i})\geq\pi_i^f({\sigma}_i,\sigma_{-i})$ for all $\sigma_{-i}$.
    It remains to show that the strict inequality $\pi_i^f(\hat{\sigma}_i,\sigma_{-i})>\pi_i^f({\sigma}_i,\sigma_{-i})$ holds for all regular $\sigma_{-i}$.

	Since $\sigma_i$ is not equivalent to its extremization $\hat{\sigma}_i$, there are two alternatives $a,b\in X$ for which there is a positive probability that $a$ and $b$ are respectively the best and worst alternatives for agent $i$, and the reported payoff of either $a$ or $b$ is not extreme. That is, $\mathbf{P}\{U_i\in E\cup F\}>0$, where
	\begin{equation*}
    \begin{split}
	E&:=\{u_i\in[0,1]^m|u_i^a=\textstyle\max_xu_i^x>\min_xu_i^x=u_i^b\text{ and }\sigma_i^b(u_i)>0\},\\
    F&:=\{u_i\in[0,1]^m|u_i^a=\textstyle\max_xu_i^x>\min_xu_i^x=u_i^b\text{ and }\sigma_i^a(u_i)<1\}.
    \end{split}
	\end{equation*}
    We have either (i) $\mathbf{P}\{U_i\in E\}>0$ or (ii) $\mathbf{P}\{U_i\in F\}>0$. We assume that (i) holds. The proof for the case when (ii) holds is similar, and hence omitted.

 To show the strict inequality part of the lemma, it suffices to prove the following claim:
	
	\bigskip
	\noindent
	\textbf{Claim A.}
	\textit{$\pi_i^f(\hat{\sigma}_i(u_i),\sigma_{-i}|u_i)>\pi^f_i(\sigma_i(u_i),\sigma_{-i}|u_i)$ for all $u_i\in E$ and regular $\sigma_{-i}$.}
	
	\bigskip
	\noindent
	If we can show Claim A, then, combined with (\ref{eq:ineq}), we will have proved that $\pi_i^f(\hat{\sigma}_i,\sigma_{-i})>\pi_i^f({\sigma}_i,\sigma_{-i})$ for all regular $\sigma_{-i}$, as required.

	To prove Claim A, fix a type $u_i\in E$ and a regular strategy subprofile $\sigma_{-i}$.
	We first give a formula of the conditional expected payoff.
	For each alternative $x\in X$, denote by 
	\[
	S^x:=\sum_{j\neq i}\lambda_j\sigma_j^x(U_j)
	\]
	the (random) weighted utilitarian sum induced by $x$, taken over the agents excluding $i$, based on the reported types.
	Agent $i$'s conditional expected payoff given his true type $u_i$ if he reports $v_i\in[0,1]^m$ can be written as:
	\begin{equation}
	    \pi_i^f(v_i,\sigma_{-i}|u_i)=\sum_{x\in X}u_i^x\mathbf{P}\{
	    \lambda_iv_i^x+S^x=\textstyle\max_{y}\lambda_iv_i^y+S^y|U_i=u_i\}.
	    \label{eq:pi0}
	\end{equation}
    We already know that (\ref{eq:pi0}) is nondecreasing in $v_i^a\in[0,1]$. It thus suffices to show that (\ref{eq:pi0}) is strictly decreasing in $v_i^b\in[0,1]$.

	The probabilities in the sum (\ref{eq:pi0}) add up to 1, since the regularity of $\sigma_{-i}$ implies that ties almost never occur in the ranking of the variables $\lambda_iv_i^x+S^x$, $x\in X$. We normalize (\ref{eq:pi0}) by
	subtracting $u_i^{b}$:
	\begin{equation}
	  \pi_i^f(v_i,\sigma_{-i}|u_i)-u_i^b=\sum_{x\neq b}(u_i^x-u_i^b)\mathbf{P}\{
	    \lambda_iv_i^x+S^x=\textstyle\max_{y}\lambda_iv_i^y+S^y|U_i=u_i\}.
	    \label{eq:pi}
	\end{equation}
	Clearly, each of the probabilities in the sum (\ref{eq:pi}) is nonincreasing in $v_i^b\in[0,1]$. Since  $u_i\in E$, we have $u_i^x\geq u_i^{b}$ for all $x$, with strict inequality for $x=a$. Thus, to show that the sum (\ref{eq:pi}) is strictly decreasing in $v_i^b$, it suffices to show that the probability in the $a$-indexed term of the sum,
	\begin{equation}
	\mathbf{P}\{
	    \lambda_i+S^a=\textstyle\max_{y}\lambda_iv_i^y+S^y|U_i=u_i\},
    \label{eq:prob}
    \end{equation}
    is strictly decreasing in $v_i^b$.

    To do this, rewrite (\ref{eq:prob}) as
    \begin{equation}
        \begin{split}
            \mathbf{P}\{
	    \lambda_i(v_i^a-v_i^x)\geq S^x-S^a\text{ for all }x\neq {a}|U_i=u_i\}
	    =G((\lambda_i(v_i^a-v_i^x))_{x\neq {a}}|u_i),
        \end{split}
        \label{eq:prob1}
    \end{equation}
    where $G(\cdot|u_i)$ is the conditional joint cumulative distribution function of the random vector $(S^x-S^a)_{x\neq{a}}=\left(\sum_{j\neq i}\lambda_j\left(\sigma_j^x(U_j)-\sigma_j^a(U_j)\right)\right)_{x\neq{a}}$ given that $U_i=u_i$. The variable $v_i^b$ affects (\ref{eq:prob1}) only via the $b$-indexed argument 
    $\lambda_i(v_i^a-v_i^b)$.
    The regularity of $(\sigma_j)_{j\neq i}$ implies that the distribution function $G(\cdot|u_i)$ has full support over $[-\sum_{j\neq i}\lambda_j,\sum_{j\neq i}\lambda_j]^{n-1}$, and hence (\ref{eq:prob1}) is strictly decreasing in $v_i^b$ as long as the condition $\lambda_i(v_i^a-v_i^b)\in[-\sum_{j\neq i}\lambda_j,\sum_{j\neq i}\lambda_j]$ is satisfied. The last condition indeed holds for all $v_i^a,v_i^b\in[0,1]$, since we have chosen $i$ so that $\lambda_i\leq\sum_{j\neq i}\lambda_j$.  
	\end{proof}

	The following proposition, combined with Proposition \ref{prop:eff}, establishes Theorem \ref{thm:main}.
	
	\begin{prop}
	A social choice function is weighted utilitarian and Bayesian incentive compatible if and only if it is dictatorial. 
	\label{prop:bic}
	\end{prop}

	\begin{proof}
	Every dictatorial SCF is obviously Bayesian incentive compatible. We use Lemma \ref{lem:dom} to show that if $f$ is a nondictatorial weighted utilitarian SCF, then $f$ is not Bayesian incentive compatible. Let $i$ be an agent to whom $f$ assigns a positive weight not greater than one-half. Assumption \ref{as:dist} implies that the truth-telling strategy $\tau_i$ is not equivalent to the extremized strategy $\hat{\tau}_i$, and that the subprofile $\tau_{-i}$ of truth-telling strategies is regular. Hence, by Lemma \ref{lem:dom}, $i$'s unilateral deviation from the truth-telling strategy profile $(\tau_j)_{j\in I}$ to $\hat{\tau}_i$ makes $i$ strictly better off.  
	\end{proof}

    \section{Interpretation of the theorem in terms of cardinal voting}
    
    This paper focuses on aggregation of cardinal preferences without monetary transfers. Such a situation is typical of voting, where monetary transfers are generally forbidden, and voters have heterogeneous preference intensities. 

A common practice for reflecting preference intensity is to use a cardinal voting rule, in which voters assign a score to each alternative, and the decision is made based on the (weighted) sum of the scores (see, e.g., \citealt{nunez2014preference}). This class includes a wide variety of evaluative voting rules. However, cardinal voting may lead to inefficient outcomes due to strategic behavior of voters. \cite{kikuchi2022winner} pointed out this problem in the context of collective decision making among distinct groups, such as states or parties. In that paper, we considered a game in which the groups collectively decide between two alternatives by cardinal weighted voting. Each group chooses an internal rule that specifies the allocation of weight to the alternatives, as a function of the preferences of group members. We showed that the game is a social dilemma: the weakly dominant strategy for each group is the winner-take-all rule, while the equilibrium is Pareto dominated.

\cite{kikuchi2022winner} considered a fixed cardinal weighted voting rule. Is it possible to solve the social dilemma inherent in cardinal voting by using any decision mechanism? The impossibility theorem (Theorem \ref{thm:main}) provides a negative answer: even if arbitrary social choice functions are allowed, ex ante Pareto efficiency and Bayesian incentive compatibility cannot be achieved simultaneously, except for the case of dictatorship. Moreover, Theorem \ref{thm:main} holds for any number of alternatives, while the result of \cite{kikuchi2022winner} is limited to the binary case.

Behind the impossibility theorem is the incentive of voters in cardinal voting to exaggerate their preference intensity.\footnote{\cite{nunez2014preference} shows that due to such exaggeration behavior of voters, different cardinal voting rules can be strategically equivalent in large elections.} By Proposition \ref{prop:eff}, ex ante Pareto efficiency is equivalent to the use of a weighted utilitarian SCF. The direct mechanism for such an SCF can be seen as weighted cardinal voting where voters assign to each alternative a score in a continuous scale. As shown by Lemma \ref{lem:dom}, voters have an incentive to give extreme scores to their best and worst alternatives, which prevents the implementation of the efficient SCF. This incentive can be most explicitly expressed in the binary case, as Corollary \ref{cor:binary} below shows. In the direct mechanism of an SCF, we say that a strategy $\sigma_i$ for an agent is \textit{weakly dominant} if it weakly dominates all strategies not equivalent to $\sigma_i$.

\begin{cor}
Suppose that the number of alternatives is $m=2$, and consider a weighted utilitarian SCF $f$ such that no agent has more than half of the total weight: $\lambda_i\leq1/2$ for all $i$. In the direct mechanism for the SCF $f$, define (uniquely up to equivalence) a strategy $\sigma^\ast_i$ by, for each type $u_i=(u_i^a,u_i^b)$,
	\[
	\sigma_i^\ast(u_i)=(\sigma_i^{\ast a}(u_i),\sigma_i^{\ast b}(u_i)):=
	\begin{cases}
	(1,0)&\text{if $u_i^a>u_i^b$}\\
	(0,1)&\text{if $u_i^a<u_i^b$}.
	\end{cases}
	\]
 The strategy $\sigma_i^\ast$ is weakly dominant for all agents $i$.
 \label{cor:binary}
\end{cor}

\begin{proof}
    The extremization of any strategy is $\sigma_i^\ast$. By Lemma \ref{lem:dom}, $\sigma_i^\ast$ weakly dominates all strategies not equivalent to it.
\end{proof}

Corollary \ref{cor:binary} is a paraphrase of the result in \cite{kikuchi2022winner} that the winner-take-all rule is the weakly dominant strategy for all groups. 
Corollary \ref{cor:binary} is also an example in which incentive compatibility entails ordinal preference aggregation, which is consistent with Lemma 2 of \cite{AzrieliKim2014}.

\section{Conclusion}

This paper provides an impossibility theorem on Pareto efficiency and Bayesian incentive compatibility in a general class of social choice problems.
The main contribution of the paper lies in the generality of the conditions under which the impossibility holds. 
Most notably, we relaxed the independence assumption imposed in \cite{EHLERS202031}, as well as the binary assumption in \cite{schmitz2012sub}. To the best of our knowledge, this is the first paper that proves the impossibility without imposing preference independence or limiting the number of alternatives.

Applying the model to cardinal voting, we obtain insight that the impossibility theorem elucidates the underlying logic behind the manipulability of cardinal voting in a general case with an arbitrary number of alternatives, as well as the dilemma structure described in \cite{kikuchi2022winner} in the case of binary voting.

Characterizing the second-best mechanism and evaluating the associated welfare loss are natural questions which follow our results. 
Another possible extension would be the consideration of interim or ex post Pareto efficiency, instead of ex ante efficiency.
However, they are beyond the scope of the current paper and we leave them for future research.

\phantomsection
\addcontentsline{toc}{section}{References}
\bibliography{bibbib}

\newcommand{\noop}[1]{}
\begin{thebibliography}{13}
\newcommand{\enquote}[1]{``#1''}
\expandafter\ifx\csname natexlab\endcsname\relax\def\natexlab#1{#1}\fi

\bibitem[\protect\citeauthoryear{Azrieli and Kim}{Azrieli and
  Kim}{2014}]{AzrieliKim2014}
\textsc{Azrieli, Y. and S.~Kim} (2014): \enquote{Pareto efficiency and weighted
  majority rules,} \emph{International Economic Review}, 55, 1067--1088.

\bibitem[\protect\citeauthoryear{Bhargava, Majumdar, and Sen}{Bhargava
  et~al.}{2015}]{bhargava2015incentive}
\textsc{Bhargava, M., D.~Majumdar, and A.~Sen} (2015):
  \enquote{Incentive-compatible voting rules with positively correlated
  beliefs,} \emph{Theoretical Economics}, 10, 867--885.

\bibitem[\protect\citeauthoryear{B\"{o}rgers and Postl}{B\"{o}rgers and
  Postl}{2009}]{BORGERS20092057}
\textsc{B\"{o}rgers, T. and P.~Postl} (2009): \enquote{Efficient compromising,}
  \emph{Journal of Economic Theory}, 144, 2057--2076.

\bibitem[\protect\citeauthoryear{Ehlers, Majumdar, Mishra, and Sen}{Ehlers
  et~al.}{2020}]{EHLERS202031}
\textsc{Ehlers, L., D.~Majumdar, D.~Mishra, and A.~Sen} (2020):
  \enquote{Continuity and incentive compatibility in cardinal mechanisms,}
  \emph{Journal of Mathematical Economics}, 88, 31--41.

\bibitem[\protect\citeauthoryear{Gibbard}{Gibbard}{1973}]{gibbard1973manipulation}
\textsc{Gibbard, A.} (1973): \enquote{Manipulation of voting schemes: a general
  result,} \emph{Econometrica: journal of the Econometric Society}, 587--601.

\bibitem[\protect\citeauthoryear{Jackson and Sonnenschein}{Jackson and
  Sonnenschein}{2007}]{jackson2007overcoming}
\textsc{Jackson, M.~O. and H.~F. Sonnenschein} (2007): \enquote{Overcoming
  incentive constraints by linking decisions 1,} \emph{Econometrica}, 75,
  241--257.

\bibitem[\protect\citeauthoryear{Kikuchi and Koriyama}{Kikuchi and
  Koriyama}{2022}]{kikuchi2022winner}
\textsc{Kikuchi, K. and Y.~Koriyama} (2022): \enquote{The winner-take-all
  dilemma,} \emph{Theoretical Economics}, forthcoming.

\bibitem[\protect\citeauthoryear{Kim}{Kim}{2017}]{kim2017ordinal}
\textsc{Kim, S.} (2017): \enquote{Ordinal versus cardinal voting rules: A
  mechanism design approach,} \emph{Games and Economic Behavior}, 104,
  350--371.

\bibitem[\protect\citeauthoryear{Majumdar and Sen}{Majumdar and
  Sen}{2004}]{majumdar2004ordinally}
\textsc{Majumdar, D. and A.~Sen} (2004): \enquote{Ordinally Bayesian incentive
  compatible voting rules,} \emph{Econometrica}, 72, 523--540.

\bibitem[\protect\citeauthoryear{Mas-Colell, Whinston, and Green}{Mas-Colell
  et~al.}{1995}]{MWG1995}
\textsc{Mas-Colell, A., M.~D. Whinston, and J.~R. Green} (1995):
  \emph{Microeconomic Theory}, Oxford University Press.

\bibitem[\protect\citeauthoryear{N{\'u}{\~n}ez and Laslier}{N{\'u}{\~n}ez and
  Laslier}{2014}]{nunez2014preference}
\textsc{N{\'u}{\~n}ez, M. and J.~F. Laslier} (2014): \enquote{Preference
  intensity representation: strategic overstating in large elections,}
  \emph{Social Choice and Welfare}, 42, 313--340.

\bibitem[\protect\citeauthoryear{Satterthwaite}{Satterthwaite}{1975}]{satterthwaite1975strategy}
\textsc{Satterthwaite, M.~A.} (1975): \enquote{Strategy-proofness and Arrow's
  conditions: Existence and correspondence theorems for voting procedures and
  social welfare functions,} \emph{Journal of economic theory}, 10, 187--217.

\bibitem[\protect\citeauthoryear{Schmitz and Tr{\"o}ger}{Schmitz and
  Tr{\"o}ger}{2012}]{schmitz2012sub}
\textsc{Schmitz, P.~W. and T.~Tr{\"o}ger} (2012): \enquote{The (sub-)
  optimality of the majority rule,} \emph{Games and Economic Behavior}, 74,
  651--665.

\end{thebibliography}

\end{document}